\newcounter{myctr}

\documentclass{article}

\usepackage{url}
\usepackage{graphicx}
\usepackage{amsmath}
\usepackage{amsfonts}
\usepackage{url}
\usepackage{fixmath}
\usepackage{bm}
\usepackage{amsmath,amssymb}
\usepackage{times}
\usepackage{algorithmic}
\usepackage{algorithm}
\usepackage{amssymb}
\usepackage{amsmath}
\usepackage{amsfonts}
\usepackage{subfigure}
\usepackage{nccmath}
\usepackage{graphicx}
\usepackage{url}
\usepackage{multirow}
\usepackage{wrapfig}
\usepackage{caption}
\usepackage{color}
\usepackage{times}
\usepackage{fixmath}
\usepackage{bm}
\usepackage{mathtools}

\usepackage{amsthm}

\newtheorem{definition}{Definition}

\newtheorem{lemma}{Lemma}
\newtheorem{theorem}{Theorem}

\begin{document}

%\makeatletter
%\def\@biblabel#1{[#1]}
%\makeatother

%\markboth{Mostafa Haghir Chehreghani}{Effectively Counting  $s-t$ Simple Paths in Directed Graphs}

%\catchline{}{}{}{}{}
%
%%%%%%%%%%%%%%%%%%%%%%%%%%%%%%%%%%%%%%%%%%%%%%%%%%%%%%%%%%%%%%%%%%%%

\title{Effectively Counting  $s-t$ Simple Paths in Directed Graphs
%\footnote{For the title, try not to
%use more than 3 lines. Typeset the title in 10 pt
%Times Roman, uppercase and boldface.}
}

\author{ Mostafa Haghir Chehreghani \\
%\footnote{Typeset names in
%10~pt Times Roman, uppercase. Use the footnote to indicate
%the present or permanent address of the author.}
Department of Computer Engineering\\
               Amirkabir University of Technology (Tehran Polytechnic), Tehran, Iran
%\footnote{State completely without abbreviations, the
%affiliation and mailing address, including country. Typeset in 8~pt
%Times italic.}
\\
mostafa.chehreghani@aut.ac.ir}

%\author{SECOND AUTHOR}
%
%\address{Group, Laboratory, Address\\
%City, State ZIP/Zone, Country\\
%second\_author@group.com}

\maketitle

%\begin{history}
%\received{(received date)}
%\revised{(revised date)}
%%\accepted{(Day Month Year)}
%%\comby{(xxxxxxxxxx)}
%\end{history}

\begin{abstract}
An important tool in analyzing complex social and information networks
is $s-t$ simple path counting, which
is known to be $\#$P-complete.
In this paper, we
study efficient $s-t$ simple path
counting in directed graphs.
For a given pair of vertices $s$ and $t$
in a directed graph, first we propose a pruning technique
that can efficiently and considerably reduce the search space.
Then, we discuss how this technique
can be adjusted with exact and approximate algorithms, to improve their efficiency.
In the end, by performing extensive experiments
over several networks from different domains,
we show high empirical efficiency
of our proposed technique.
Our algorithm is not a competitor of   existing methods,
rather, it is a friend that can be
used as a fast pre-processing step,
before applying any existing algorithm.
\end{abstract}


\section{Introduction}
\label{sec:introduction}
Graphs (networks) are powerful tools
that are used to model data in different
domains,
including social networks, information networks, road networks and the world wide web.
A property seen in most of these real-world networks is that the links between the vertices do not always represent reciprocal relations \cite{Newman03thestructure}.
Consequently, the generated networks are {\em directed graphs},
where any edge has a direction and the edges are not necessarily symmetric.

{\em Path counting} is an important problem in analyzing
large networks.
For example, some of network indices, such as
{\em $k$-path centrality}
are based on counting paths \cite{DBLP:conf/cikm/ChehreghaniBA19}.
A computationally demanding path counting problem is
% An important problem in analyzing graph data is
{\em $s-t$ simple path counting}.
In this problem, given two vertices $s$ and $t$,
the goal is to count the number of
{\em \underline{all simple paths}} from $s$ to $t$.
A \textit{simple} path is path that contains each node
of the graph at most once.
For both directed and undirected graphs,
this problem is known to be
$\#$P-complete \cite{DBLP:journals/siamcomp/Valiant79}\footnote{It is important to not confuse
our studied problem with the
problem wherein
the constraint "simple" is removed,
i.e., we count the number of all paths (simple or non-simple) between two vertices $s$ and $t$.
In this problem, a node may appear for several times in a path. Unlike our studied problem (which is $\#$P-complete),
this problem can be efficiently solved in polynomial time,
using e.g., dynamic programming or matrix multiplication.
Another relevant problem,
that unlike our studied problem
can be solved in polynomial time,
is counting the number of shortest paths between
two vertices.
}.
Hence, it is critical to develop algorithms that work efficiently in practice.
% One solution to overcome this computational barrier was to develop approximate algorithm
% that work considerably faster than exact algorithms,
% in the expense of generating approximate results \cite{DBLP:journals/jgaa/RobertsK07}.

In this paper, we focus on efficient $s-t$ simple path
counting in directed graphs.
Our key contributions are
as follows.
\begin{itemize}
\item
First,
for a given pair of vertices $s$ and $t$,
we propose a pruning technique
that can reduce the search space significantly.
Moreover, it can be computed very efficiently.
\item
Then, we discuss how this technique
can be adjusted with exact and approximate algorithms, to improve their performance.
We show, for example, that under some conditions, applying this pruning technique
yields a polynomial time exact algorithm for enumerating $s-t$ simple paths.
% Under these conditions,
% existing algorithms remain still intractable.
\item
Finally, by performing extensive experiments
over several networks from different domains
(social, peer-to-peer, communication, citation,
stack exchange, product co-purchasing, ...),
we show high empirical efficiency
of our pruning technique.
\end{itemize}

Our algorithm is not a competitor of
existing $s-t$ simple path counting algorithms.
Rather, it is a friend as it can be
used as a fast pre-processing step,
before applying any of existing algorithms.

The rest of this paper is organized as follows.
In Section~\ref{sec:preliminaries},
we introduce
preliminaries and necessary definitions used in the paper.
In Section~\ref{sec:relatedwork}, we give an overview on related work.
In Section~\ref{sec:counting},
we present our pruning technique
and discuss how it improves exact and approximate algorithms.
In Section~\ref{sec:experimentalresults}, we empirically evaluate our pruning technique.
Finally, the paper is concluded in Section~\ref{sec:conclusion}.

\section{Preliminaries}
\label{sec:preliminaries}

% In this section, we present definitions and notations widely used in the paper.
We assume that the reader is familiar with basic concepts in graph theory.
Throughout the paper, $G$ refers to a directed graph.
For simplicity, we assume that $G$ is a connected and loop-free graph without multi-edges.
By default, we assume that $G$ is an unweighted graph,
unless it is explicitly mentioned that $G$ is weighted.
$V(G)$ and $E(G)$ refer to the set of vertices and the set of edges of $G$, respectively.
% For a vertex $v \in V(G)$, by $G \setminus v$ we refer to the set of connected graphs
% generated by removing $v$ from $G$.
For a vertex $v \in V(G)$,
the number of head ends adjacent to $v$
is called its {\em in degree} and
the number of tail ends adjacent to $v$ is called its {\em out degree}.
% We denote {\em out degree} of vertex $v$ in graph $G$ with $\mathcal{ON}_G(v)$.
For a vertex $v$, by $\mathcal{ON}_G(v)$ we denote the set of outgoing neighbors of $v$
in the graph $G$.
A (directed) {\em walk} in a directed graph is a sequence of edges directed in the same direction and joins a sequence of vertices.
A directed {\em trail} is a (directed) walk
wherein all the edges are distinct.
A (simple and directed) {\em path}
is a (directed) trail wherein all the vertices are distinct.
Let $A \subset V(G)$.
The subgraph of $G$ {\em induced} by $A$
is the graph whose vertices are $A$ and
edges are all the edges in $E(G)$ that have both endpoints in $A$.

% Table \ref{table:notation} summarizes symbols and notations used in the paper.
%
% \begin{table}
% \caption{Symbols and their definitions. \label{table:notation}}
% \centering
% \begin{tabular}{ p{3cm} p{7cm} }
% \hline
% Symbol & Definition \\
% \hline
% $G$ & A directed graph \\
% $V(G)$ & The set of vertices of $G$ \\
% $E(G)$ & The set of edges of $G$ \\
% % $n$ & The number of vertices of $G$ \\
% % $m$ & The number of edges of $G$ \\
% $\mathcal{ON}_{G}(v)$ & The set of outgoing neighbors of $v$ in $G$. \\
% $s,t$ & A pair of vertices \\
% $\mathcal{S}_{st}$ & The set of vertices that are in the scope of $s$ with respect to $t$ \\
% $\mathcal{IS}_{ts}$ & The set of vertices that are in the inverse scope of $t$ with respect to $s$ \\
% $\rho_{st}(G)$ & The pruned graph (with respect to $s$ and $t$) \\
% $\mathcal X$ & The set of simple paths
% from $s$ to $t$ \\
% $\mathbf X$ & The sample state of the algorithm of \cite{DBLP:journals/jgaa/RobertsK07} \\
% $X$ & Estimation of $|\mathcal X|$ \\
% $h(p)$ & The probability of sampling path $p$ in our approximate algorithm \\
% $h'(p)$ & The probability of sampling path $p$ in the algorithm of \cite{DBLP:journals/jgaa/RobertsK07} \\
% $\mathbb{E}[\cdot]$ & Expected value of $\cdot$ \\
% $\mathbb{V}ar[\cdot]$ & Variance of $\cdot$ \\
% \hline
% \end{tabular}
% \end{table}
%

%%%%%%%%%%%%%%%%%%%%%%%%%%%%%%%%%%%%%%%%%%%%%%%%%%%%%%%%%%%%%%%%%%%%%%%%%%%%%%%%%%%%%%

\section{Related work}
\label{sec:relatedwork}

There exist a number of algorithms in the literature for counting/enumerating/listing $s-t$ simple paths.
In one of the earlier studies,
Valiant \cite{DBLP:journals/siamcomp/Valiant79}
showed that it is $\#$P-complete
to enumerate simple $s-t$ paths
in both directed and undirected graphs.
In \cite{HURA1983157},
Hura utilized {\em Petri nets} to enumerate all
$s-t$ simple paths in a directed graph.
Bax \cite{DBLP:journals/ipl/Bax94} exploited Hamiltonian path algorithms to derive an
$O\left(2^{|V(G)|} poly(|V(G)|)\right)$ time
and $O\left(poly(|V(G)|)\right)$ space algorithm
for $s-t$ simple path counting.
Knuth \cite{Knuth_4a} presented the {\em simpath} algorithm that generates a
(not-always-reduced)
binary decision diagram for all simple paths from $s$ to $t$.
The constructed diagram is called
{\em zero-suppressed decision diagram}
({\em ZDD}).
In recent years, more efficient algorithms
have been proposed to construct {\em ZDD} \cite{DBLP:conf/ambn/YasudaSM17}.
Roberts and Kroese
\cite{DBLP:journals/jgaa/RobertsK07}
presented an approximate algorithm for
$s-t$ simple path counting,
which is based on {\em sequential importance sampling} (we discuss
their algorithm in more details in Section~\ref{sec:approximate},
and explain how our pruning technique can improve it).
In \cite{DBLP:journals/mst/MihalakSW16},
Mihal{\'{a}}k et.al. studied
approximately counting approximately shortest paths in directed graphs.
However, their algorithm is restricted to
directed acyclic graphs (DAGs).

In the literature,
there are problems that are close to our studied problem,
or are a restricted form of it.
In the following, we briefly review them:
\begin{itemize}
\item
In a restricted form of our studied problem,
the number of $s-t$ simple paths of
a (fixed) length $l$ are counted.
Flum and Grohe
\cite{DBLP:journals/siamcomp/FlumG04} proved that
counting simple paths
of length $l$ on both directed and undirected graphs,
parameterized by $l$, is \#W[1]-complete.
Recently,
Giscard et.al.
\cite{DBLP:journals/algorithmica/GiscardKW19}
proposed an algorithm for counting
$s-t$ simple paths of length up to $l$.
Time complexity of their algorithm is
$O\left( |V(G)| + |E(G)| + (l^w+l\Delta) |S_l| \right)$,
where
$\Delta$ is the maximum degree of the graph,
$|S_l|$ is the number of weakly connected induced subgraphs of $G$ on at most $l$ vertices,
and $w$ is the exponent of matrix multiplication.
\item
The other problem which is close
to our studied problem is counting the number of all paths
from a vertex $s$ to another vertex $t$.
However, unlike our studied problem
(which is $\#$P-complete),
this problem
can be solved in polynomial time,
using e.g., dynamic programming
% (yielding $O(m+n)$ time algorithms)
or fast matrix multiplication.
% (yielding $O(n^{1.})$ time algorithms).
For example and as described in e.g., \cite{jrnl:Brandes},
the number of all paths of size $k$
(simple or non-simple)
from $s$ to $t$ is equal to the $st$-th entry
of $A^k$, where $A$ is the adjacency matrix of the graph
and $A^k$ is the $k$-th power of $A$.
% Let Ak = (a
% (k)
% uv )u,v∈V be the k-th
% power of the adjacency matrix of an unweighted graph. Then a
% (k)
% uv equals the
% number of paths from u to v of length exactly k.
% They first build a directed ascyclic grapn (DAG) $D$
% rooted at $s$.
% \begin{equation}
% \text{Paths}(u) =
% \begin{cases}
%     1 & \text{if } u = t \\
%     \sum_{(u,v) \in E(D)} \text{Paths}(v) & \text{otherwise.}\\
% \end{cases}
% \end{equation}
Arenas et.al. \cite{DBLP:conf/pods/ArenasCJR19}
study the problem of counting
% to retrieve
the number of all paths between two vertices $s$ and $t$
of a length at most $k$,
and show that it admits a fully polynomial-time randomized approximation scheme (FPRAS).
\item
Another relevant problem
is counting the number of shortest paths between
a given pair of vertices.
Unlike our studied problem,
this problem can
be solved efficiently in polynomial time, too.
Using breadth first search (BFS), this problem
can be solved in $O\left(|V(G)|+|E(G)|\right)$ time \cite{jrnl:Brandes}.
Recently, a number of faster approximate algorithms
have been proposed to address this problem
in static \cite{DBLP:journals/algorithms/MensahGY20}
and dynamic graphs \cite{DBLP:conf/cikm/TretyakovAGVD11}.
The notion of shortest paths is used to develop
several tools of network analysis, including
{\em betweenness centrality} \cite{DBLP:journals/cj/Chehreghani14,DBLP:journals/corr/abs-1708-08739,DBLP:conf/edbt/ChehreghaniAB19}.
\end{itemize}

In the current paper, we propose a pruning technique
to improve exact/approximate algorithms
of the {\em general form of $s-t$ simple path counting} in directed graphs.
Our algorithm is not a competitor of these
existing algorithms.
Rather, it is a friend that can be
used as a fast pre-processing step before applying any of them.

\section{Counting $s-t$ paths in directed graphs}
\label{sec:counting}

In this section,
given a source vertex $s$ and a target vertex $t$
in a directed graph $G$,
we propose algorithms to
count the number of all simple
paths from $s$ to $t$.
First in Section~\ref{sec:prunning},
we present a pruning technique
used to reduce search space for
$s-t$ path counting.
Then in Section~\ref{sec:exact},
we discuss how this pruning technique
improves exact $s-t$ path counting.
Finally in Section~\ref{sec:approximate},
we study how our pruning technique
can improve approximate algorithms.

\subsection{A pruning technique}
\label{sec:prunning}

In this section, we describe our pruning technique.
We start with introducing the sets
$\mathcal{S}_{st}$ (Definition~\ref{def:rf}) and
$\mathcal{IS}_{ts}$
(Definition~\ref{def:tf}),
that are used to reduce the search space.

\begin{definition}
\label{def:rf}
Let $G$ be a directed graph and
$s,t,v \in V(G)$ such that $s\neq t$.
We say $v$ is {\em in the scope of $s$ with respect to $t$} iff
either $v=s$, or there is
at least one directed path
from $s$ to $v$ in $G$ that
does not pass over $t$.  \footnote{The path may end to $t$, however, it can not pass over $t$.}
The set of vertices $v$ that are
{\em in the scope of $s$ with respect to $t$} is denoted with $\mathcal{S}_{st}$.
\end{definition}

\begin{definition}
\label{def:inverse}
Let $G$ be a directed graph.
{\em Inverse graph} of $G$,
denoted with $\mathcal I(G)$,
is a directed graph such that:
(i) $V(\mathcal I(G))=V(G)$, and
(ii) $(u,v) \in E(\mathcal I(G))$ if and only if $(v,u) \in E(G)$ \cite{DBLP:journals/corr/abs-1708-08739}.
\end{definition}

\begin{definition}
\label{def:tf}
Let $G$ be a directed graph and $s,t,v \in V(G)$ such that $s \neq t$.
We say $v$ is {\em inversely in the scope of $t$ with respect to $s$} iff
either $v=t$, or there is
at least one directed path from $t$ to $v$
in $\mathcal I(G)$
that does not pass over $s$  \footnote{The path may end to $s$, however, it can not pass over $s$.}.
The set of vertices $v$ that are {\em inversely in the scope of $t$ with respect to $s$} is denoted with $\mathcal{IS}_{ts}$.
\end{definition}

For example, in the graph of Figure~\ref{fig:approximate},
$\mathcal{S}_{st}$ consists of vertices
$s$, $v_1$, $v_2$, $v_3$, $v_4$, $v_5$ and
$t$.
It does not contain vertex $v_6$,
as the path from $s$ to $v_6$
passes over $t$.
Moreover,
$\mathcal{IS}_{ts}$ consists of vertices
$t$, $v_5$ and $s$.
For a given pair of vertices $s,t \in V(G)$,
we can compute $\mathcal S_{st}$
and $\mathcal {IS}_{ts}$
efficiently in $O(|E(G)|)$ time.
% $\mathcal{RF}(r)$ and $\mathcal{RT}(r)$ can be efficiently computed,
% using {\em inverse graph}.
% To compute $\mathcal{S}_{st}$,
% we act as follows \cite{DBLP:journals/corr/abs-1708-08739}:
To compute $\mathcal{S}_{st}$,
we act as follows:
\begin{enumerate}
\item
First, if $G$ is weighted,
weights of the edges of $G$
are discarded,
\item
Then, a (revised) breadth first search (BFS) or a depth-first search (DFS) on $G$
starting from $s$ is conducted,
with an small change:
when $t$ is met,
the traversal is not expanded to
the children (and descendants) of $t$.
All the vertices that are met during the traversal are appended to $\mathcal{S}_{st}$.
\end{enumerate}

To compute $\mathcal{IS}_{ts}$, we act as follows:
\begin{enumerate}
\item
First, by flipping the direction of the edges of $G$,
we construct $\mathcal I(G)$,
\item
Then, if $G$ is weighted,
weights of the edges are discarded,
\item
Finally, a (revised) breadth-first search or
depth-first search on $\mathcal I(G)$
starting from $t$ is conducted,
with an small change:
when $s$ is met, the traversal is not
expanded to its children (and descendants).
All the vertices that are met during the traversal are appended to
$\mathcal{IS}_{ts}$.
\end{enumerate}

It is easy to see that both
$\mathcal{S}_{st}$ and $\mathcal{IS}_{ts}$ can be computed in $O(|E(G)|)$ time,
for both unweighted and weighted graphs.
Furthermore, we have the following lemma.

\begin{lemma}
\label{lemma:exactbc}
Given a directed graph $G$ and
vertices $s,t \in V(G)$,
the number of simple paths from $s$ to $t$
is equal to
the number of simple paths from
$s$ to $t$ whose vertices belong to
$\mathcal{S}_{st} \cap \mathcal{IS}_{ts}$.
\end{lemma}
\begin{proof}
Each path should start from $s$ and end with $t$,
therefore its vertices should belong to both
$\mathcal{S}_{st}$ and
$\mathcal{IS}_{ts}$.
\end{proof}

Lemma~\ref{lemma:exactbc} says that
in order to compute $s-t$ paths,
we require to consider only those paths
that their vertices belong to both
$\mathcal{S}_{st}$ and
$\mathcal{IS}_{ts}$.
We can use this to prune many vertices of the graph that do not belong to
either $\mathcal{S}_{st}$ or
$\mathcal{IS}_{ts}$ or both.
As a result, if in a graph,
$\mathcal{S}_{st} \cap \mathcal{IS}_{ts}$
is (considerably) smaller than the number of vertices of the graph,
this pruning technique
can enormously improve the efficiency
of the $s-t$ path counting algorithm.
Note that time complexity of
$s-t$ simple path counting algorithms
is exponential in terms of the number
of vertices of the graph.
Hence, discarding a considerable part of the vertices can hugely improve
the efficiency of the algorithms.

In Algorithm~\ref{alg:framework},
we use this optimization technique
to make $s-t$ path counting algorithms
more efficient.
If $out\_degree$ of $s$ is $0$ or
$in\_degree$ of $t$ is $0$,
there is no path from $s$ to $t$, hence,
Algorithm~\ref{alg:framework} returns $0$.
Then, it computes $\mathcal{S}_{st}$
and $\mathcal{IS}_{ts}$ and stores them respectively
in the sets $S$ and $IS$.
Then, for each vertex $v$ in the graph,
we check if it belongs to both $S$ and $IS$.
We form the subgraph induced by all such vertices and
call it $\rho_{st}(G)$.
Finally, we apply the $s-t$ path counting algorithm
on $\rho_{st}(G)$.
the graph $\rho_{st}(G)$ is usually much smaller than $G$, therefore,
the algorithm can be run much faster.

\begin{figure}
\caption{\label{fig:approximate}
The graph on left
is a directed graph $G$ and
the graph on right
% Figure~\ref{fig:approximate_1}
% shows a directed graph $G$ and
% Figure~\ref{fig:approximate_2}
shows $\rho_{st}(G)$.
Over $G$ and $s$ and $t$,
while the variance of our algorithm
is $0$, the variance of the algorithm
of \cite{DBLP:journals/jgaa/RobertsK07} is $4$.}
\centering
\subfigure[A graph $G$]
{
\includegraphics[scale=0.35]{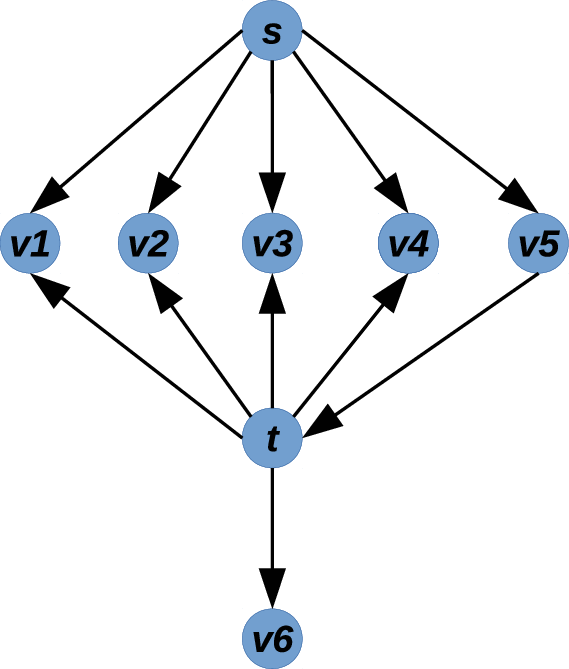}
\label{fig:approximate_1}
}
\hspace{4em}
\subfigure[$\rho_{st}(G)$]
{
\includegraphics[scale=0.35]{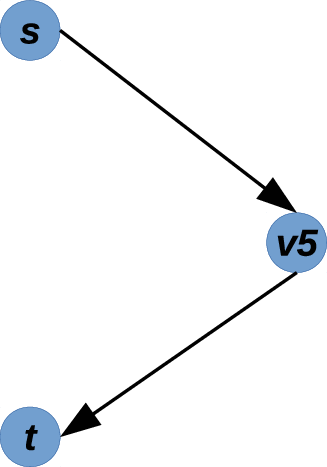}
\label{fig:approximate_2}
}
\end{figure}

\begin{algorithm}
\caption{High level pseudo code of
computing the number of $s-t$ paths in a directed graph.}
\label{alg:framework}
\begin{algorithmic} [1]
\STATE \textbf{Input.} A directed network $G$,
a pair of vertices $s,t \in V(G)$.
\STATE \textbf{Output.}
The number of simple paths from $s$ to $t$.
\IF{{\em out degree} of $s$ is $0$ or {\em in degree} of $t$ is $0$}
\RETURN $0$.
\ENDIF
\STATE $S \leftarrow$ compute
$\mathcal{S}_{st}$.
\STATE $IS \leftarrow$ compute $\mathcal{IS}_{ts}$.
\FORALL{vertices $v \in V(G)$}
\IF{$v \in S$ and $v \in IS$}
\STATE Mark $v$.
\ENDIF
\ENDFOR
\STATE $\rho_{st}(G) \leftarrow$ the induced subgraph
of $G$ consisting of the marked vertices.
\STATE $path\_no \leftarrow$ call an
(exact/approximate) $s-t$ path counting algorithm on
$\rho_{st}(G)$.
\RETURN $path\_no$.
\end{algorithmic}
\end{algorithm}

\subsection{Exact algorithm}
\label{sec:exact}

In an exact algorithm for counting $s-t$ simple paths that exhaustively enumerates all simple paths from $s$ to $t$, we can use {\em backtracking}:
we start from $s$, take a path
and walk it
(so that the path does not contain repeated vertices).
If the path ends to $t$, we count it and backtrack and take some other path.
If the path does not reach to $t$,
we discard it and take some other path.
This algorithm enumerates all the paths from $s$ to $t$.
We refer to this algorithm as
{\em Exhaustive Simple Path Enumerator}, or \textsf{ESPE} in short.
As discussed in \cite{DBLP:journals/siamcomp/Valiant79},
enumerating $s-t$ paths for both directed and
undirected graphs is $\#$P-complete.
% This means the enumerating paths $p_1$
In the following, we discuss that
using our pruning technique and
for a certain type of vertices in a directed graph,
% and under some conditions,
% if our  is applied,
% the above-mentioned method becomes tractable, i.e.,
the simple path enumeration problem can be solved in
% it enumerates all the simple paths in
{\em polynomial time}.

% To study whether or not an enumeration problem
% can be solved efficiently,
% the authors of \cite{DBLP:journals/ipl/JohnsonP88} presented three notions for  enumeration/listing complexity.
% Let $S$ be a list with $k$ elements.
% Its elements, $S_1,S_2,\ldots,S_k$
% are listed with:
% \begin{itemize}
% \item
% {\em polynomial delay},
% if the time before enumerating (printing) $S_1$,
% the time between enumerating $S_i$ and
% $S_{i+1}$ for every $i=1,\ldots,k-1$,
% and the termination time after enumerating
% $S_k$ is bounded by
% a polynomial of the size of the input,
% \item
% {\em incremental polynomial time}
% if the time between enumerating
% $S_i$ and $S_{i+1}$ for every
% $i=1,\ldots,k-1$
% (respectively the termination time after printing $S_k$) is bounded by a polynomial
% of the combined sizes of the input and
% $\{S_1,\ldots,S_i\}$
% (respectively $\{S_1,\ldots,S_k\}$),
% \item
% {\em output polynomial time}
% (or {\em polynomial total time}) if $S$ is enumerated in the combined sizes
% of the input and $S$.
% \end{itemize}
%
% Obviously,
% {\em polynomial delay}
% is better than {\em incremental polynomial time} and {\em incremental polynomial time}
% is better than
% {\em output polynomial time}.
% Moreover, if a problem is $\#$P-complete,
% then it can not be solved with
% {\em polynomial delay}.
% % In the rest of this section, we show that
% % (under some reasonable assumption)
% % in real-world directed graphs,
% % the $s-t$ path counting problem can be solved
% % with {\em polynomial delay}.

In the $s-t$ simple path counting problem,
the input size consists of
the number of vertices and the number of edges of $G$.
In Theorem~\ref{theorem:polynomial_delay}
we show that if for $s$ and $t$,
$\left|\mathcal S_{st} \cap \mathcal {IS}_{ts}\right|$
is a constant, then the $s-t$ simple counting problem
can be solved efficiently.
Later in Section~\ref{sec:experimentalresults}
we empirically show that in most of real-world
networks, vertices $s$ and $t$
have usually a small quantity for
$\left|\mathcal S_{st} \cap \mathcal {IS}_{ts}\right|$.
It should be highlighted that even for such pairs of vertices,
applying the above mentioned exhaustive simple path enumeration algorithm without
using our pruning technique,
{\em does not yield a  polynomial time algorithm}.

\begin{theorem}
\label{theorem:polynomial_delay}
Let $G$ be a directed graph, where
for a given pair of vertices $s$ and $t$,
$|V\left(\rho_{st}(G)\right)|)$ is a constant.
The list of simple paths
from $s$ to $t$ can be enumerated in polynomial time
(in terms of $|V(G)|$ and $|E(G)|$).
\end{theorem}
\begin{proof}
We can compute the sets
$\mathcal{S}_{st}$,
$\mathcal{IS}_{ts}$,
and their intersection in
$O(|E(G)|)$ time.
So, $\rho_{st}(G)$ can be computed in
$O(|E(G)|)$ time.
After that, we work with $\rho_{st}(G)$
which has only a constant number of vertices.
Hence, all the paths can be enumerated
in a constant time.
\end{proof}

In fact,
the proof of Theorem~\ref{theorem:polynomial_delay} presents a much better result than polynomial time:
it says that under the mentioned condition
and after a linear time spent to
construct $\rho_{st}(G)$,
the whole list can be enumerated in a
constant time (in terms of $|V(G)|$ and $|E(G)|$).

\subsection{Approximate algorithm}
\label{sec:approximate}

Roberts and Kroese \cite{DBLP:journals/jgaa/RobertsK07} proposed
randomized algorithms for estimating
the number of $s-t$ simple paths in a graph.
In this section, we investigate how the first randomized algorithm of \cite{DBLP:journals/jgaa/RobertsK07} can be improved
using our pruning technique \footnote{The second randomized algorithm
of \cite{DBLP:journals/jgaa/RobertsK07}
and most of the other randomized algorithms
in the literature can be improved in a similar way.}.
First, we briefly describe
how this algorithm changes if we apply
our pruning technique on it.
Then, we analyze this revised algorithm and
compare it with the first algorithm of \cite{DBLP:journals/jgaa/RobertsK07}.

In the approximate algorithm, we sample $N$ independent (random) paths
% , along with computing
and compute the probability of sampling each path.
The following procedure is used to sample each path:
\begin{enumerate}
\item
Start with vertex $s$.
Initialize the following variables:
$c \leftarrow s$ (current vertex),
$h \leftarrow 1$ (probability),
$k \leftarrow 1$ (counter), and
$p \leftarrow \{s\}$ (path).
\item
Mark $s$ as $visited$ (to be sure that $p$ will not visit $s$ again).
% $A(\cdot·, 1)=0$ to ensure that the path will not return to $1$.
\item
Let $V'= \{v \in \mathcal{ON}_{\rho_{st}(G)}(c) | v \text{ is not } visited\}$
be the set of possible vertices for the next step of the path.
% ($V'= \{v \in V | A(c,v) = 1\}$).
\item
Choose the next vertex $v \in V'$ of the path uniformly at random
and set $x_{k+1}$ to $v$.
\item
Set $c \leftarrow v$, $h \leftarrow \frac{h}{|V'|}$,  $k \leftarrow k+1$, and $p \leftarrow p \cup \{v\} $.
Mark $v$ as $visited$.
\item
If $c=t$, then stop.
Otherwise go to step $3$.
\end{enumerate}

If we do not use our pruning technique,
instead of working with $\rho(G)$,
we should work (e.g., in step 3) with $G$.
This has consequences.
For example, while a path sampled
from $G$ may never reach to $t$
(so after some iterations, $V'$ becomes empty),
this never happens when
we work with $\rho(G)$.
Hence, as stated in \cite{DBLP:journals/jgaa/RobertsK07},
when working with $G$ we need to check the
following condition in step 3:
if $V'= \emptyset$,
we do not generate a valid path, so stop.

Let $x_1,x_2,\ldots,x_m$
be the set of vertices of each path $p$,
with $x_1=s$.
In the above mentioned procedure,
path $p$ is chosen with probability:
\begin{align}
\label{eq:prob}
h(p) =& h(x_2| x_1 ) \cdot h(x_3|x_1,x_2 ) \ldots
h(x_m|x_1,\ldots, x_{m-1}) \nonumber \\
     =& \left|\mathcal{ON}_{\rho_{st}(G)}(x_1) \setminus \{x_1\}\right| \cdot
     \left|\mathcal{ON}_{\rho_{st}(G)}(x_2) \setminus \{x_1,x_2\}\right| \ldots \nonumber \\
      & \left|\mathcal{ON}_{\rho_{st}(G)}(x_{m-1}) \setminus \{x_1,\ldots,x_{m-1}\}\right|
\end{align}

Let $\mathcal X$ be the set of simple paths from $s$ to $t$.
Let also $p^1,\ldots,p^N$ be the set of sampled paths.
For each
% sampled path
$p^k$, we estimate the number of
$s-t$ simple paths as $X^k=\frac{1}{h(p^k)}$.
The final estimation is the average of
estimations of different samples:
\begin{equation}
X = \frac{1}{N} \sum_{k=1}^N \frac{1}{h(p^k)}.
\end{equation}

It is easy to see that $X$ is an unbiased estimator
for the number of simple paths from $s$ to $t$:
$$\mathbb{E}[X^k]=\sum_{p \in \mathcal X} \frac{1}{h(p)} \cdot h(p) = \sum_{p \in \mathcal X} 1 = |\mathcal X|.$$
Since $X$ is the average of $N$ random variables
whose expected values is $|\mathcal X|$,
the expected value of $X$ is $|\mathcal X|$, too.

For the variance of $X$, we have:
\begin{align*}
\mathbb Var \left[ X^k \right] &=
\mathbb E\left[(X^k)^2 \right] -
\mathbb E\left[X^k\right]^2 \nonumber  \\
   &= \sum_{p \in \mathcal X} \frac{1}{h(p)^2} \cdot h(p) - \mathcal X^2
   = \sum_{p \in \mathcal X}\frac{1}{h(p)} - \mathcal X^2 .
\end{align*}
Since $X$ is the average of $N$ independent
random variable $X^1,\ldots, X^N$,
we have:
\[\mathbb Var \left[ X\right]
= \frac{1}{N} \sum_{p \in \mathcal X} \frac{1}{h(p)}
- \frac{\mathcal X^2}{N}.\]

Now let compare this variance with the variance of
the estimator of the algorithm of \cite{DBLP:journals/jgaa/RobertsK07}.
% Let $p^1,\ldots,p^N$ be the set
% of paths sampled by the algorithm of \cite{DBLP:journals/jgaa/RobertsK07}.
Unlike our algorithm, in the algorithm of \cite{DBLP:journals/jgaa/RobertsK07}
it is possible that a sampled path does not belong to
$\mathcal X$.
Let $\mathbf{X}$ be the sample space of
the algorithm of \cite{DBLP:journals/jgaa/RobertsK07} that consists of those paths
in $G$ that start from $s$ and end to $t$ or to a vertex whose $V'$ is empty.
In a way similar to our analysis, its variance is:
\begin{align}
\label{eq:variance2}
\mathbb Var \left[ X\right] &= \frac{1}{N}
\sum_{p \in \mathbf X} \frac{I(p \in \mathbf X)}{h'(p)}
- \frac{\mathcal X^2}{N} \nonumber \\
&= \frac{1}{N}
\sum_{p : p \in \mathcal X} \frac{1}{h'(p)}
- \frac{\mathcal X^2}{N},
\end{align}
where
$h'(p)$ is defined as follows
($x_1,\ldots,x_m$ are vertices of $p$):
\begin{align}
\label{eq:prob2}
h'(p) =& \left|\mathcal{ON}_{G}(x_1) \setminus \{x_1\}\right| \cdot
     \left|\mathcal{ON}_{G}(x_2) \setminus \{x_1,x_2\}\right| \ldots \nonumber \\
      & \left|\mathcal{ON}_{G}(x_{m-1}) \setminus \{x_1,\ldots,x_{m-1}\}\right|.
\end{align}
Moreover in Equation~\ref{eq:variance2},
$I(p \in \mathcal X)$ is an indicator
function, which is $1$ if
$p$ is a path from $s$ to $t$,
and $0$ otherwise.

Comparing Equation~\ref{eq:prob} with
Equation~\ref{eq:prob2}
reveals that our algorithm has a better (lower) variable.
In order to have a lower variance,
an algorithm should assign as low as possible probabilities to the paths that
it may sample but they
do not belong
to $\mathcal X$.
Because in this case their inverse will be a
very large number and the contribution of these large numbers to the variance will be $0$.
Our algorithm assigns probability $0$ to such
paths, so it manages them very efficiently.
However, the algorithm of \cite{DBLP:journals/jgaa/RobertsK07}
may assign a (large) non-zero value to many of them.
For example, in Figure~\ref{fig:approximate}
assume that we want to estimate the number of simple paths from vertex $s$ to vertex $t$.
The sample space of our algorithm
consists of only only one path which is the path
$s \rightarrow v_5 \rightarrow t$.
Hence, the variance of its estimation,
even when $N$ is $1$, is $0$.
In contrast, the sample space of the algorithm
of \cite{DBLP:journals/jgaa/RobertsK07}
consists of $5$ paths:
$s \rightarrow v_1$,
$s \rightarrow v_2$,
$s \rightarrow v_3$,
$s \rightarrow v_4$, and
$s \rightarrow v_5 \rightarrow t$,
where it assigns an equal probability $1/5$ to each path.
Among them, only the last one belongs to $\mathcal X$.
While this algorithm assigns the same
probability to each one, our algorithm gives probability $0$ to the paths that do not belong
to $\mathcal X$.
With $N=1$, the variance of the estimation of
the algorithm of \cite{DBLP:journals/jgaa/RobertsK07}
is $4$.
Since our algorithm assigns probability $0$
to the paths that are not in $\mathcal X$,
it never finds a variance worse than the algorithm of \cite{DBLP:journals/jgaa/RobertsK07}.

\begin{table*}[!htb]
\caption{\label{table1}
Specifications of our large networks.}
\begin{center}
\resizebox{\textwidth}{!}{% use resizebox with textwidth
\begin{tabular}{ l | l l l }
\hline
Dataset  & \#vertices & \#edges &  domain \\
\hline
% \hline
soc-sign-Slashdot090221 \cite{DBLP:conf/chi/LeskovecHK10}   & {82140} & 549202 & social \\
soc-sign-epinions \cite{DBLP:conf/chi/LeskovecHK10} & {131828} & 841372 &  social \\
p2p-Gnutella08 \cite{DBLP:journals/tkdd/LeskovecKF07}   & {6301} & 20777  & peer-to-peer file sharing \\
p2p-Gnutella30 \cite{DBLP:journals/tkdd/LeskovecKF07} & 36682 & 88328 & peer-to-peer file sharing \\
p2p-Gnutella31 \cite{DBLP:journals/tkdd/LeskovecKF07}  & 62586 & 147892 & peer-to-peer file sharing \\
wikiVote  \cite{DBLP:conf/chi/LeskovecHK10}    & {7115} & 103689 & Wikipedia vote \\
collegeMsg     \cite{DBLP:journals/jasis/PanzarasaOC09} & {1899} & 59835  & message exchanging \\
soc-Epinions1 \cite{DBLP:conf/semweb/RichardsonAD03}   & {75879} & 508837 & who-trust-whom \\
email-EuAll   \cite{DBLP:journals/tkdd/LeskovecKF07}  & {265214} & 420045 & communication \\
cit-HepPh   \cite{DBLP:conf/kdd/LeskovecKF05}     & {34546} & 421578 & citation \\
cit-HepTh   \cite{DBLP:conf/kdd/LeskovecKF05}      & {27770} & 352807 & citation \\
% amazon0505        & \url{https://snap.stanford.edu/data/amazon0505.html} & {410236} & 3356824 & product co-purchasing \\
sx-mathoverflow  \cite{DBLP:conf/wsdm/ParanjapeBL17}  & {24818} & 506550 & stack exchange (Math Overflow) \\
sx-askubuntu  \cite{DBLP:conf/wsdm/ParanjapeBL17}     & {159316} & 964437 & stack exchange (Ask Ubuntu) \\
sx-superuser  \cite{DBLP:conf/wsdm/ParanjapeBL17}    & {194085} & 1443339 & stack exchange (Super User) \\
% amazon0601        & \url{https://snap.stanford.edu/data/amazon0601.html} & {403394} & 3387388 & product co-purchasing \\
amazon0302    \cite{DBLP:journals/tweb/LeskovecAH07}    & {262111} & 1234877 &  product co-purchasing \\
amazon0312    \cite{DBLP:journals/tweb/LeskovecAH07}    & {400727} & 3200440 & product co-purchasing \\
\hline
\end{tabular}
}
\end{center}
\end{table*}

% \section{Dynamic graphs}
% \label{sec:dynamic}
%
%
% In many applications, the generated graphs
% are {\em dynamic}, i.e., they change over time
% by a sequence of {\em update operations}.
% An update operation might be a vertex insertion,
% or a vertex deletion, or an edge insertion,
% or an edge deletion.
% Moreover, if the graph is weighted,
% it might also be an edge weight change.
% In this section, we study the applicablity
% and usefulness of our proposed framework
% for $s-t$ path counting in dynamic graphs.
%
%
% After an update operation in $G$,
% the first step is to update the sets
% $\mathcal{RT}(s)$ and $\mathcal{RF}(t)$.
% Updating each of these sets is explained in
% \cite{DBLP:conf/bigdataconf/ChehreghaniBA18a}
% for betweenness centrality estimation.
% The second step is to update the set of samples.
% Let $N$ be the number of samples.
% To do so, with probability $\frac{N-1}{N}$
%
% for each sample

\section{Experimental results}
\label{sec:experimentalresults}

In order to evaluate the empirical behaviour of our pruning technique,
we perform extensive experiments
over several real-world networks
from different domains.
The objective is to empirically show that
in real-world directed networks,
for a given pair of vertices $s$ and $t$,
$\mathcal{S}_{st} \cap \mathcal{IS}_{ts}$
is small (compared to $V(G)$),
and it can be computed very efficiently.
These two automatically improve efficiency
of any (exact/approximate) simple path
counting algorithm.

\begin{table}[!htb]
\caption{\label{table2}Empirical evaluation
of our proposed pruning technique.
Column $max_{sz}$
presents the maximum size of
$\mathcal{S}_{st} \cap \mathcal{IS}_{ts}$
divided by the number of vertices
(over different pairs of $s,t$ vertices).
Column $max_{tm}$ presents
the maximum time to compute
$\mathcal{S}_{st} \cap \mathcal{IS}_{ts}$,
where the reported times are in second(s).
}
\begin{center}
\begin{tabular}{ p{5cm} | p{5cm}  p{1.5cm} }
\hline
Dataset & $max_{sz}$ & $max_{tm}$ \\
\hline
% \hline
soc-sign-Slashdot090221 &   0.3333 & 0.0889  \\
soc-sign-epinions &  0.3144 & 0.1172 \\
p2p-Gnutella08    &   0.3296 &  0.0014 \\
p2p-Gnutella31    &    0.2261 & 0.0150 \\
p2p-Gnutella30    &    0.2315 & 0.0081 \\
wikiVote          &   0.1570 & 0.0068\\
collegeMsg        &   0.6821 & 0.0004 \\
email-EuAll       &  0.1289 &  0.0618 \\
soc-Epinions1     &  0.4246 &  0.0761 \\
cit-HepPh         &  0.0000002 & 0.2833 \\
cit-HepTh         &  0.0000002 & 0.2765 \\
sx-mathoverflow   &  0.1478 &  0.0495 \\
sx-askubuntu      &  0.1160 &  0.1399 \\
sx-superuser      &  0.1477 &  0.2418 \\
% amazon0601        &  {403394} & 3387388 &  {13} & 0.284 \\
amazon0302        &  0.9229 &   0.2255 \\
amazon0312        &  0.9488 &   0.5525 \\
% amazon0505        &  {410236} & 3356824 &  {17} & 0.224 \\
\hline
\end{tabular}
\end{center}
\end{table}

\begin{table*}[!htb]
\caption{\label{table4}
Specifications of our smaller networks.}
\begin{center}
\begin{tabular}{ l | l l l }
\hline
Dataset & \#vertices & \#edges &  domain \\
\hline
bio-CE-GN \cite{DBLP:conf/aaai/RossiA15} & 2220 & 53683 & Gene functional associations \\
bio-SC-LC \cite{DBLP:conf/aaai/RossiA15} & 2004 & 20452 & Gene functional associations \\
bn-mouse\_retina\_1 \cite{DBLP:conf/aaai/RossiA15} & 1123 & 90811 & Brain network  \\
C1000-9 \cite{DBLP:conf/aaai/RossiA15}  & 1001 & 450081 & DIMACS \\
C2000-9 \cite{DBLP:conf/aaai/RossiA15}  & 2001 & 1799534 & DIMACS \\
SmaGri  \cite{DBLP:conf/aaai/RossiA15}  & 1060 & 4921 & Citation network \\
\hline
\end{tabular}
\end{center}
\end{table*}

\begin{table*}[!htb]
\caption{\label{table5}Empirical comparison
of our proposed pruning technique.
All the reported times are in seconds
and "not terminated" means the algorithm does not terminate
within 12 hours!
}
\begin{center}
\begin{tabular}{ l | c c c  c }
\hline
Dataset & max $\left|\mathcal{S} \cap \mathcal{IS}\right|$ & \multicolumn{2}{c}{max total running time} & max $\#$simple paths \\
\cline{3-4}
&   & \textsf{pruned-ESPE} &  \textsf{ESPE} &  \\
\hline
bio-CE-GN  &   135 & \textbf{5.253}   & not terminated & 1,180,840   \\
bio-SC-LC &   144 & \textbf{414.605} & not terminated & 119,738,817 \\
bn-mouse\_retina\_1  & 29 & \textbf{183.987} & not terminated & 17,012,976 \\
C1000-9   & 35    & \textbf{195.333} & 448.225 & 919,343,373     \\
C2000-9  & 34    & \textbf{171.182} & 253.073 & 887,306,870  \\
SmaGri  & 52    & \textbf{0.001} & 5.973 & 1,030   \\
\hline
\end{tabular}
\end{center}
\end{table*}

The experiments are done on an Intel
processor clocked at $2.4$ GHz with $3$ GB
main memory, running Ubuntu Linux $18.04.2$ LTS.
% The program is compiled by the GNU C++ compiler 5.4.0 using optimization level 3.
We perform our tests over 16 real-world datasets
from the SNAP repository\footnote{\url{https://snap.stanford.edu/data/}}.
The specifications of the datasets
are summarized in Table~\ref{table1}.
Over each dataset,
we choose $500$ pairs of vertices $s$ and $t$,
uniformly at random.
Then, we run our proposed pruning technique
for each pair.
In the end, we report maximum size of
$\mathcal{S}_{st} \cap \mathcal{IS}_{ts}$
divided by $|V(G)|$
(in the $max_{sz}$ column)
and the maximum time to compute
$\mathcal{S}_{st} \cap \mathcal{IS}_{ts}$
(in the $max_{tm}$ column),
over all the sampled pairs.

The results are reported in Table~\ref{table2}.
On the one hand,
since $max_{sz}$ is usually considerably smaller than the number of vertices of the graph, our pruning technique usually makes
the real-world graphs considerably smaller.
Hence, any $s-t$ path counting problem
(exact, approximate, ...) can be solved
much faster on the smaller graph.
Compared to the number of vertices,
we may consider $\mathcal{S}_{st} \cap \mathcal{IS}_{ts}$
as a constant.
Therefore after pruning,
the $s-t$ path counting algorithm finds a constant time.
Note that the reported values for
$|\mathcal{S}_{st} \cap \mathcal{IS}_{ts}|/|V(G)|$
are the maximum ones and most of the other values are much smaller.

On the other hand,
the values reported in the $max_{tm}$ column
indicate that computing
$\mathcal{I}_{st}$,
$\mathcal{IS}_{ts}$, and their intersection
can be done very quickly.
Note that $max_{tm}$ contains
all the time required to compute
$\mathcal{S}_{st}$,
$\mathcal{IS}_{ts}$, and their intersection.
As reported in the table,
over all the datasets,
$max_{tm}$ is less than $0.6$ seconds!
This time is quite ignorable,
compared to the time of counting all $s-t$ simple paths.
Note that while worst case time complexity
of our pruning technique is $O(|E(G)|)$,
in practice it is usually
performed much faster.
The reason is that for most of vertices,
the direction of edges
makes the sets $\mathcal{S}$ and $\mathcal{IS}$ small.
As a result, BFSs (or DFSs)
conducted on these small sets,
can be conducted very efficiently.

In our experiments,
we tried to compare the \textsf{ESPE} algorithm
(discussed in Section~\ref{sec:exact}),
against the method wherein first we apply our pruning technique and then we use \textsf{ESPE}
(we refer to this method as \textsf{pruned-ESPE}).
% We tried to run the exact algorithm of $s-t$ path counting over the datasets.
However, over any of the above mentioned networks,
none of the methods \textsf{ESPE} and \textsf{pruned-ESPE}
finish during a reasonable time (24 hours!).
% even for small datasets such as collegeMsg.
Therefore and to compare the efficiency of \textsf{pruned-ESPE}
against \textsf{ESPE}, we switch to smaller graphs.
% for which at least one algorithm can produce the result within
% a reasonable time.
Our smaller real-world graphs have between 1K-2.3K vertices.
Their specifications are summarized in Table~\ref{table4}
(they can be downloaded from \url{http://networkrepository.com/networks.php}). We treat all these networks as directed graphs.
% More information  can be found in \cite{DBLP:conf/aaai/RossiA15}.

When comparing \textsf{pruned-ESPE} against \textsf{ESPE},
we perform the following scenario:
over each dataset, we sample uniformly at random 500 pairs of nodes $s$ and $t$, such that there exists at least one path from $s$ to $t$.
Then for each pair, we run \textsf{pruned-ESPE} and \textsf{ESPE}
and compute the following statistics:
$\left|\mathcal{S} \cap \mathcal{IS}\right|$,
total running time of \textsf{pruned-ESPE}
(which includes the pre-processing/pruning phase),
running time of \textsf{ESPE}, and the number of simple paths.
In the end and among all the sampled pairs,
we {\em choose} the pair that yields
the maximum number of simple paths
and report the results of this {\em chosen} pair.
For both algorithms,
this {\em chosen} pair usually takes the most amount of time
to process.
Thus, "max $\left|\mathcal{S} \cap \mathcal{IS}\right|$"
depicts the value of $\left|\mathcal{S} \cap \mathcal{IS}\right|$
for the {\em chosen} pair, "max total running time" depicts
the total time to process the {\em chosen} pair (by each of the algorithms),
and "max \#simple paths" presents the number of simple paths between
the two nodes of the {\em chosen} pair.

Table~\ref{table5} reports the results.
As seen in the table, using our proposed pruning techniques
significantly improves the running time.
All the times reported in this table are in seconds,
and "not terminated" means the algorithm does not terminate
within 12 hours!
So, for the {\em chosen} pair and over bio-CE-GN, bio-SC-LC and bn-mouse\-retina\-1, while \textsf{ESPE} doe not produce the results
during a reasonable time (12 hours!),
\textsf{pruned-ESPE} processes the {\em chosen} pairs quickly.
Over the other graphs that both algorithms terminate within a reasonable time, \textsf{pruned-ESPE} acts quite faster.
As evidenced by the values of the "max $\left|\mathcal{S} \cap \mathcal{IS}\right|$" column,
this is due to the efficiency of our pruning technique in  considerably reducing the search space.

\section{Conclusion}
\label{sec:conclusion}

In this paper, we studied efficient
$s-t$ simple path counting in directed graphs. For a given pair of vertices $s$ and $t$
in a directed graph, first we presented a pruning technique that can efficiently reduce the search space.
Then, we investigated how this
pruning technique can be adjusted with exact and approximate algorithms, to improve their efficiency.
In the end, we performed extensive
experiments over several real-world networks to show the empirical efficiency of our proposed pruning technique.

% Acknowledgments---Will not appear in anonymized version
% \acks{We thank a bunch of people.}

% \nocite{*}
% \bibliographystyle{fundam}
 \bibliographystyle{plain}
\bibliography{allpapers}   % name your BibTeX data base

\end{document}